\documentclass[11pt]{amsart}
\usepackage{fullpage}
\usepackage[foot]{amsaddr}
\usepackage{microtype}
\usepackage[OT1]{fontenc}
\usepackage{eulervm}
\usepackage{amsfonts}
\usepackage{amsmath}
\usepackage{amssymb}
\usepackage{amsthm}
\usepackage{mathtools} 
\usepackage[linesnumbered,boxed,ruled,vlined]{algorithm2e}
\usepackage{algpseudocode}
\usepackage{enumitem}
\usepackage{bm}
\usepackage{bbm}
\usepackage{xifthen}
\usepackage{xspace}

\usepackage[margin=1cm]{caption} 
\usepackage{subfig}

\usepackage[thinlines]{easytable}

\usepackage[bookmarks=true,hypertexnames=false,pagebackref]{hyperref}
\hypersetup{colorlinks=true, citecolor=blue, linkcolor=red, urlcolor=blue}

\usepackage{pgfplots}
\pgfplotsset{compat=1.16}
\usepackage{tikz}
\usetikzlibrary{arrows,arrows.meta,backgrounds,calc,fit,decorations.pathreplacing,decorations.markings,shapes.geometric}

\tikzstyle{internal} = [draw, fill, shape=circle]
\tikzstyle{external} = [shape=circle]
\tikzstyle{square}   = [draw, fill, rectangle]
\tikzstyle{triangle} = [draw, fill, regular polygon, regular polygon sides=3, inner sep=3pt]
\tikzstyle{pentagon} = [draw, fill, regular polygon, regular polygon sides=5, inner sep=2pt, minimum size=14pt]
\tikzset{every fit/.append style=text badly centered}

\usetikzlibrary{positioning,chains,fit,shapes,calc}
\usetikzlibrary{trees}
\usetikzlibrary{decorations.pathreplacing}
\usetikzlibrary{decorations.pathmorphing}
\usetikzlibrary{decorations.markings}
\tikzset{>=latex} 

\usepackage{ifthen}

\usepackage{cleveref}

\usepackage[textsize=tiny]{todonotes}

\usepackage[normalem]{ulem}

\usepackage{mleftright}

\usepackage{cool}
\Style{DSymb={\mathrm d},DShorten=true,IntegrateDifferentialDSymb=\mathrm{d}}

\renewcommand{\Pr}{\mathop{\mathrm{Pr}}\nolimits}

\def\*#1{\mathbf{#1}}
\def\+#1{\mathcal{#1}}
\def\-#1{\mathrm{#1}}
\def\=#1{\mathbb{#1}}

\newcommand{\abs}[1]{\ensuremath{\left\vert#1\right\vert}}

\newcommand{\dist}{\operatorname{dist}}

\newcommand{\defeq}{:=}

\newtheorem{theorem}{Theorem}

\newtheorem{lemma}[theorem]{Lemma}

\newtheorem{proposition}[theorem]{Proposition}

\theoremstyle{definition}

\newtheorem{definition}[theorem]{Definition}

\theoremstyle{remark}

\crefname{theorem}{Theorem}{Theorems}
\crefname{observation}{Observation}{Observations}
\crefname{claim}{Claim}{Claims}
\crefname{condition}{Condition}{Conditions}
\crefname{algorithm}{Algorithm}{Algorithms}
\crefname{property}{Property}{Properties}
\crefname{example}{Example}{Examples}
\crefname{fact}{Fact}{Facts}
\crefname{lemma}{Lemma}{Lemmas}
\crefname{corollary}{Corollary}{Corollaries}
\crefname{definition}{Definition}{Definitions}
\crefname{remark}{Remark}{Remarks}
\crefname{proposition}{Proposition}{Propositions}
\crefname{equation}{equation}{equations}
\crefname{enumi}{Case}{Case}
\creflabelformat{enumi}{(#2#1#3)}

\makeatletter
\def\prob#1#2#3{\goodbreak\begin{list}{}{\labelwidth\z@ \itemindent-\leftmargin
      \itemsep\z@  \topsep6\p@\@plus6\p@
      \let\makelabel\descriptionlabel}
  \item[\textbf{Name}]#1
  \item[\textbf{Instance}]#2
  \item[\textbf{Output}]#3
  \end{list}}
\makeatother

\makeatletter
\providecommand\@dotsep{5}
\def\listtodoname{Todo list}
\def\listoftodos{\@starttoc{tdo}\listtodoname}
\makeatother

\newcommand{\IntLength}{\ensuremath{\frac{1+\delta}{\Delta}}}
\newcommand{\Vol}{\operatorname{Vol}}
\newcommand{\Inc}{\operatorname{Inc}}
\newcommand{\Flat}{\operatorname{Flat}}

\newcommand{\zo}[1]{[0,1]^{#1}}     
\usepackage{bigints}
\usepackage{nicefrac,comment}

\newboolean{doubleblind}
\setboolean{doubleblind}{false}

\title{Deterministic approximation for the volume of the truncated fractional matching polytope}

\ifdoubleblind
\author{}
\else

\author{Heng Guo, Vishvajeet N}

\address[]{School of Informatics, University of Edinburgh, Informatics Forum, 10 Crichton Street, Edinburgh, EH8 9AB, UK.}
\email{hguo@inf.ed.ac.uk}
\email{nvishvajeet@gmail.com}

\thanks{This project has received funding from the European Research Council (ERC) under the European Union's Horizon 2020 research and innovation programme (grant agreement No.~947778).}

\fi

\begin{document}

\begin{abstract}
  We give a deterministic polynomial-time approximation scheme (FPTAS) for the volume of the truncated fractional matching polytope for graphs of maximum degree $\Delta$,
  where the truncation is by restricting each variable to the interval $[0,\IntLength]$, and $\delta\le \frac{C}{\Delta}$ for some constant $C>0$.
  We also generalise our result to the fractional matching polytope for hypergraphs of maximum degree $\Delta$ and maximum hyperedge size $k$, truncated by $[0,\IntLength]$ as well,
  where $\delta\le C\Delta^{-\frac{2k-3}{k-1}}k^{-1}$ for some constant $C>0$.
  The latter result generalises both the first result for graphs (when $k=2$), and a result by Bencs and Regts (2024) for the truncated independence polytope (when $\Delta=2$).
  Our approach is based on the cluster expansion technique.
\end{abstract}
\maketitle

\section{Introduction}

The evaluation of volume is a classic computing task.
The pioneer work of Dyer, Frieze, and Kannan \cite{DFK91} showed that given a membership oracle, the volume of a convex body can be approximated to arbitrary accuracy in randomised polynomial time.
On the other hand, in the same model, deterministic approximation takes at least exponential time \cite{Ele86,BZ87}.
This marked an early milestone that distinguished the computational power of deterministic algorithms from that of randomized algorithms.

Nonetheless, computing volumes might still be interesting for other models and/or more restricted cases,
where the lower bound of \cite{Ele86,BZ87} does not apply,
and interesting deterministic volume algorithms may still exist.
In particular, there appears to be no run-time lower bound for deterministically approximating the volume of a polytope given by facets.
In this case, as the membership oracle is trivial to implement,
the randomised algorithm of Dyer, Frieze, and Kannan \cite{DFK91} achieves an $\varepsilon$-approximation in time polynomial in the input size and $\frac{1}{\varepsilon}$.
The deterministic counterpart of such an algorithm is called \emph{fully polynomial-time approximation scheme} (FPTAS).
Unfortunately, the best deterministic algorithmic technique for polytope volumes appears to fall short of achieving FPTAS.
They either require exponential time \cite{Law91,Bar93}\footnote{The algorithms of \cite{Law91,Bar93} are in fact exact algorithms and the dominating term in their run-time is linear in the number of vertices of the polytope.} or yield exponential approximation \cite{Bar09,BR21}.
In the hope of getting an FPTAS, we may focus on even more restricted cases, such as the independence polytope or matching polytope for graphs.

Along this direction, Gamarnik and Smedira \cite{GS23} considered the relaxed independence polytope.
Given a graph $G=(V,E)$, it is defined as $\{\mathbf{x}\in\zo{V}\mid \forall\{u,v\}\in E,~x_u+x_v\le 1\}$.
Clearly, if we further restrict every $x_v$ to the interval $[0,1/2]$, this polytope degenerates to a cube of side length $1/2$ and its volume is simply $2^{-\abs{V}}$.
Gamarnik and Smedira \cite{GS23} showed that we can push beyond this trivial case,
namely, if we truncate the relaxed independence polytope by the interval $\left[ 0,\frac{1}{2}\left( 1+\frac{O(1)}{\Delta^2} \right) \right]$ for graphs of maximum degree $\Delta$,
a quasi-polynomial-time approximation scheme exists.
Their technique is based on the correlation decay approach \cite{Wei06,BG06}.
Subsequently, Benec and Regts \cite{BR24} improved the interval to $\left[ 0,\frac{1}{2}\left( 1+\frac{O(1)}{\Delta} \right) \right]$ and the run-time to polynomial-time,
based on the zeros of polynomial approach \cite{Bar16,PR17}.
The bound on this interval appears to be where the limit of the current methods is.

In this paper, we explore what other polytopes may admit efficient deterministic volume approximation.
In particular, we consider the natural dual of the independence polytope, namely the matching polytope, or more precisely,
its standard relaxation, the fractional matching polytope.
Note that although it is the dual of the relaxed independence polytope,
its volume might be drastically different.

\begin{definition}[Fractional matching polytope]\label{def:fractional-matching}
  Given a graph $G=(V,E)$, the \emph{fractional matching polytope} is defined as follows
  \begin{align}
    P_{G}\defeq\left\{ \mathbf{x} \in \zo{E}\mid \sum_{e \sim v} x_e \leq 1 \text{ for every } v \in V\right\},\notag
  \end{align}
  where $e\sim v$ if the edge $e$ is adjacent to $v$.
\end{definition}

For the fractional matching polytope, the trivial truncation is with the interval $\left[0,\frac{1}{\Delta}\right]$.
Similar to \cite{GS23,BR24}, we also truncate $P_G$ by an interval that is slightly longer, multiplicatively, than the trivial truncation.
\begin{definition}[Truncated fractional matching polytope]\label{def:truncate-fractional-matching}
  For $\delta >0$ and a graph $G=(V,E)$ of maximum degree $\Delta$, the \emph{truncated fractional matching polytope} is defined as follows
  \begin{align}
    P_{G,\delta} \defeq\left\{ \mathbf{x} \in \left[0,\IntLength\right]^E\mid \sum_{e \sim v} x_e \leq 1 \text{ for every } v \in V\right\}.\notag
  \end{align}
\end{definition}
Denote the interval $\left[0, \IntLength\right]$ by $M_{\delta}$. 
Then $P_{G,\delta}={M_{\delta}}^E\cap P_{G}$.
Additionally, for any $v\in V$, the constraint  $\sum_{e \sim v} x_e \leq 1$ is denoted by $C_v$.
We are interested in computing the volume of $P_{G,\delta}$, i.e.~the quantity 
\begin{align}\label{eqn:vol}
  \Vol(P_{G,\delta}) \defeq \bigintssss_{{M_\delta}^E} \prod _{v \in V}\mathbbm 1_{C_v}  d\mu, 
\end{align}
where $\mu$ is the Lebesgue measure and the indicator function $\mathbbm 1_{C_v}:{M_\delta}^E \rightarrow \{0,1\}$ outputs $1$ if and only if the constraint $C_v$ is satisfied, i.e.~when  $\sum_{e\sim v} x_e \leq 1$.

Our main result regarding $\Vol(P_{G,\delta})$ is the following.
Interestingly, similar to \cite{BR24}, the relative margin for the truncation interval we get is also $\frac{C}{\Delta}$.

\begin{theorem}  \label{thm:FPTAS-matching}
  For graphs of maximum degree $\Delta\ge 2$ and $\delta\le\frac{C}{\Delta}$ for some constant $C>0$,
  there is a fully polynomial-time approximation scheme (FPTAS) for $\Vol(P_{G,\delta})$.
\end{theorem}

The proof of \Cref{thm:FPTAS-matching} relies on the algorithmic cluster expansion approach \cite{HPR20,JKP20}.
This is very closely related to the zeros of polynomial approach Bencs and Regts \cite{BR24} took. 
The first step is to rewrite $\Vol(P_{G,\delta})$ as the partition function of a polymer model on the graph $G$.
Roughly speaking, each polymer corresponds to a set of connected vertices where the constraints on these vertices are violated.
For the algorithmic approach of the polymer model to work, 
we need to show that the total weights of these violations decay rapidly,
or more technically, the so-called Koteck\'{y}-Preiss criterion \cite{KP86} holds.
Intuitively, when we truncate with the trivial interval $\left[ 0,\frac{1}{\Delta} \right]$, 
none of the violation may happen.
With a longer truncation interval, violations can happen, but the small relative margin we allow implies that violations can only happen with small probability.
Our method involves some careful estimates of these violation probabilities.
We also note that it appears difficult to apply the correlation decay method here, 
and a direct application of the argument of Bencs and Regts \cite{BR24} would yield a smaller margin $\delta=\frac{C}{\Delta^2}$ for some constant~$C>0$. 

As mentioned earlier, the volumes of the relaxed independence polytope and the fractional matching polytope do not appear to be related,
and thus \Cref{thm:FPTAS-matching} is not directly comparable to the main result of \cite{BR24}.
To put both results under a unified framework,
we also consider hypergraph matchings.
Let $H=(V,E)$ be a hypergraph.
Consider the following polytope:
\begin{align}\label{eqn:hyper-polytope}
  P_{H}\defeq\left\{ \mathbf{x} \in \zo{E}\mid \sum_{e \ni v} x_e \leq 1 \text{ for every } v \in V\right\},
\end{align}
where $e\ni v$ if the hyperedge $e$ contains $v$.
This is a relaxation of the hypergraph matching polytope and a natural generalisation of the fractional matching polytope for graphs in \Cref{def:fractional-matching}.
We denote it by the \emph{fractional hypergraph matching polytope}.

We consider hypergraphs of maximum degree $\Delta\ge 2$ and maximum hyperedge size $k\ge 2$.
In this case, the polytope $P_H$ is indeed defined by a set of linear constraints where each variable appear at most $k$ times, and each constraint has at most $\Delta$ variables.
When $k=2$, this degenerates to the fractional matching polytope, and when $\Delta=2$, this degenerates to the relaxed independence polytope.
Similar to \Cref{def:truncate-fractional-matching}, we truncate it by a cube of side length $\IntLength$.
Recall that $M_{\delta}$ denotes the interval $\left[0, \IntLength\right]$.
Define the truncated polytope $P_{H,\delta}\defeq{M_{\delta}}^E\cap P_{H}$.
Our result regarding $\Vol(P_{H,\delta})$ is the following. 

\begin{theorem}  \label{thm:FPTAS-hyper-matching}
  For hypergraphs of maximum degree $\Delta\ge 2$ and maximum hyperedge size $k\ge 2$, and $\delta\le\frac{C}{\Delta^{\frac{2k-3}{k-1}}k}$ for some constant $C>0$,
  there is an FPTAS for $\Vol(P_{H,\delta})$.
\end{theorem}

The bound $\delta =O\left( \left(\Delta^{\frac{2k-3}{k-1}}k\right)^{-1}\right)$ recovers \Cref{thm:FPTAS-matching}.
Namely, when $k=2$, $\delta=O\left( \frac{1}{\Delta} \right)$.
Moreover, when $\Delta=2$, $\delta=O\left( \frac{1}{k} \right)$ and the interval length is $\frac{1+\delta}{\Delta}=\frac{1}{2}+O\left( \frac{1}{k} \right)$,
which recovers the result of~\cite{BR24}.

The proof of \Cref{thm:FPTAS-hyper-matching} also relies on the algorithmic cluster expansion approach.
It combines our technique for proving \Cref{thm:FPTAS-matching} and a generalisation of the technique in \cite{BR24}.
In particular, we introduce a different polymer model in this case,
which requires a new generalisation of the classic broken circuit theory \cite{Whi32,Tut54} to hypergraphs.

An immediate open problem is if \Cref{thm:FPTAS-hyper-matching} holds with $\delta=O\left( \frac{1}{\Delta k} \right)$, 
which is discussed in more detail in \Cref{sec:discuss}.
A much more challenging question is if we can obtain deterministic volume approximation for these polytopes without truncation,
or with truncation by the interval $[0,1-\delta]$ for some small~$\delta$.
For the closely related problems of approximating the number of independent sets or matchings,
deterministic algorithms \cite{Wei06,PR17,BGKNT07} match their randomised counterparts, at least for bounded degree graphs.
In contrast, volume approximation with no or little truncation appears to be out of reach for current deterministic methods.

The rest of the paper is organised as follows.
In \Cref{sec:matching}, we consider the fractional matching polytope and show \Cref{thm:FPTAS-matching}.
In \Cref{sec:hyper-matching}, we turn our attention to the hypergraph case and show \Cref{thm:FPTAS-hyper-matching}.
Finally, in \Cref{sec:discuss}, we discuss some potential future directions and the current obstacles.

\section{Cluster Expansion} \label{sec:matching}
In this section, we briefly review the polymer model and the algorithmic cluster expansion approach,
and apply it to $\Vol(P_{G,\delta})$ to show \Cref{thm:FPTAS-matching}.
Let $\+P$ be a finite set, whose elements we call \emph{polymers}. We endow $\mathcal P$ with a symmetric and reflexive incompatibility relation $\not \sim$ between any two polymers, and also a weight function $w(\cdot)$\footnote{The weight function can be complex-valued, but in this paper we will focus on real-valued weight functions.} that assigns a weight $w(\gamma)$ to the polymer~$\gamma$. Given two polymers $\gamma_1$ and $\gamma_2$, we write $\gamma_1\not \sim \gamma_2$ if $\gamma_1$ and $\gamma_2$ are incompatible, and $\gamma_1\sim \gamma_2$ otherwise. 
For a set $\Gamma$ of polymers, we say it is \emph{compatible} if for any two $\gamma_1,\gamma_2\in \Gamma$, $\gamma_1\sim\gamma_2$.
Additionally, there is a size function for polymers, and we use $\abs{\gamma}$ to denote the size of $\gamma$.

\begin{definition}[Polymer partition function]\label{def:polymer-model}
  The partition function of the polymer model above is
  \begin{align}
    \Xi(\mathcal P, w)\defeq \sum_{\Gamma}\prod_{\gamma \in \Gamma} w(\gamma),\notag
  \end{align}
  where the sum is over compatible $\Gamma\subseteq\+P$.
\end{definition}
We will be interested in the cluster expansion which is an infinite series representation of $\log \Xi (P)$.   Given a multiset $\Gamma$ of polymers from $\mathcal P$, we define the incompatibility graph $H(\Gamma)$ where we have a vertex $v_\gamma$ for every polymer $\gamma \in \Gamma$ and an edge between every pair of vertices corresponding to an incompatible pair of polymers. A multiset $\Gamma$ is called a \emph{cluster} if the incompatibility graph $H(\Gamma)$ is connected. We also denote the set of all clusters from $\mathcal P$ as $\mathcal C$. 
\begin{definition}[Cluster Expansion]\label{def:cluster-exp}
  The cluster expansion of $\Xi$ is 
  \begin{align}
    \log \Xi(\mathcal P, w)\defeq \sum_{\Gamma  \in \mathcal  C}\Phi(\Gamma)\prod_{\gamma \in \Gamma } w(\gamma), \notag
  \end{align}
  where $\Phi(\Gamma)$ is the Ursell function.
\end{definition}
The exact form of the Ursell function will not be important to our need, and we refer the interested reader to, for example, the survey by Jenssen \cite{Jen24}. 
The cluster expansion is a good approximation to the logarithm of the partition function under various conditions.
One of the most well-known conditions is the following.

\begin{proposition}[Koteck\'{y}-Preiss Criterion \cite{KP86}]\label{prop:kp}
  Let $g:\mathcal P\rightarrow [0,\infty)$ be a ``decay function''. 
  Suppose that for all $\gamma \in \mathcal P$, we have,
  \begin{align}\label{eqn:KP}
    \sum_{\gamma' \not \sim \gamma}\abs{w(\gamma')}  e^{\abs{\gamma'}+g(\gamma')}	\leq \abs{\gamma}.
  \end{align}
  Then, the cluster expansion converges absolutely.
\end{proposition}
We note that the condition in \Cref{prop:kp} is not the most general, but this particular form will be convenient to our use later.

We are going to recast \eqref{eqn:vol} into a polymer partition function.
Let $G=(V,E)$ be a graph of maximum degree $\Delta$.
For $S\subseteq V$, let $\+K(S)$ be the set of connected components in the induced subgraph $G[S]$.

Denote by $\overline {C_v}$ the constraint $\sum_{e\sim v} x_e > 1$.
Using the indicator function for the constraint $\overline{C_v}$, and the fact that $\mathbbm 1_{C_v}= 1-\mathbbm 1_{\overline{C_v}}$, we can expand~\eqref{eqn:vol} as
\begin{align}
  \Vol(P_{G,\delta}) &= \bigintssss_{{M_\delta}^E} \prod _{v \in V}(1-\mathbbm 1_{\overline{C_v}})d\mu\notag\\
  &=\bigintssss_{{M_\delta}^E} \sum_{S\subseteq V}(-1)^{|S|}\prod _{v \in S}\mathbbm 1_{\overline{C_v}}d\mu\notag\\
  &=\bigintssss_{{M_\delta}^E} \sum_{S\subseteq V}\prod_{K \in \+K(S)}(-1)^{\abs{K}}\prod_{v \in K}\mathbbm 1_{\overline{C_v}}d\mu. \notag
\end{align}
We then continue by pushing the integral inside:
\begin{align}
  \Vol(P_{G,\delta})
  & = \sum_{S\subseteq V}\bigintssss_{{M_\delta}^E}\prod_{K \in \+K(S)}(-1)^{\abs{K}}\prod_{v \in K}\mathbbm 1_{\overline{C_v}}d\mu\notag\\
  & = \sum_{S\subseteq V}
  \bigintssss_{{M_\delta}^{E\setminus E(S)}}1d\mu
  \left(\prod_{K \in \+K(S)}  (-1)^{\abs{K}} \bigintssss_{{M_\delta}^{E(K)}}\prod _{v \in K}\mathbbm 1_{\overline{C_v}}d\mu\right),\label{eqn:swap-integral-prod}
\end{align}
where, for a subset $S\subseteq V$ of vertices, $E(S)$ denotes the set of edges adjacent to it.
Note that we swapped the integral with the product in \eqref{eqn:swap-integral-prod}.
This is valid because the constraint regarding each $S_i$ rely on disjoint set of edges.
There is also an extra factor of $\bigintssss_{{M_\delta}^{E\setminus E(S)}}1d\mu$ to account for the edges that are not adjacent to $S$.
As $\bigintssss_{{M_\delta}^{E\setminus E(S)}}1d\mu=\left( \IntLength \right)^{\abs{E}-\abs{E(S)}}$, we have
\begin{align}
  \Vol(P_{G,\delta}) & = \left( \IntLength \right)^{\abs{E}-\abs{E(S)}} \sum_{S\subseteq V}
  \prod_{K \in \+K(S)} \bigintssss_{{M_\delta}^{E(K)}} (-1)^{\abs{K}}\prod _{v \in K}\mathbbm 1_{\overline{C_v}}d\mu \notag\\
  & = \left( \IntLength \right)^{\abs{E}} \sum_{S\subseteq V}\prod_{K \in \+K(S)} w(K),
  \label{eqn:pf-polymer}
\end{align}
where 
\begin{align} \label{eqn:polymer-weight}
  w(K) = \left( \IntLength \right)^{-\abs{E(K)}} (-1)^{\abs{K}} \bigintssss_{{M_\delta}^{E(K)}} \prod _{v \in K}\mathbbm 1_{\overline{C_v}}d\mu.
\end{align}
The redistribution of the $\left( \IntLength \right)^{-\abs{E(S)}}$ factor is valid because $\{E(K)\}_{K\in\+K(S)}$ is a partition of $E(S)$.
 
To fit \eqref{eqn:pf-polymer} into \Cref{def:polymer-model},
we call a subset $S\subseteq V$ a polymer if the induced subgraph $G[S]$ is connected.
Two polymers $S_1$ and $S_2$ are compatible if $\dist_{G}(S_1,S_2)\ge 2$, namely when they are not adjacent.
The size of a polymer is simply the number of vertices it contains.
In addition, equipped with the weight function in \eqref{eqn:polymer-weight}, we then have that
\begin{align}\label{eqn:Xi-PF}
  \Xi(G) = \frac{\Vol(P_{G,\delta})}{\left( \IntLength \right)^{\abs{E}}}.
\end{align}

The polymer model is useful thanks to the following theorem by Jenssen, Keevash, and Perkins \cite{JKP20} (building upon the works of \cite{BHKK08,PR17,HPR20}).

\begin{proposition}[\protect{\cite[Theorem 8]{JKP20}}]\label{prop:fptas-JKP20}
Fix an integer $\Delta>0$ and let $\mathcal G$ be a class of graphs of maximum degree at most  $\Delta$.
Suppose the following conditions hold for a given polymer model with decay function $g(\cdot)$ as in \Cref{prop:kp}:
\begin{enumerate}
  \item \label{item:JKP-1}
    There exist constants $c_1,c_2>0$ such that given a connected subgraph $\gamma$, determining whether $\gamma$ is a polymer, and then computing $w(\gamma)$ and $g(\gamma)$ can be done in time $O(|\gamma|^{c_1}e^{c_2|\gamma|})$.
  \item \label{item:JKP-2}
      there exists $\rho>0$ such that for every $G \in \mathcal G$ and every polymer $\gamma \in \mathcal P(G)$, $g(\gamma) \geq \rho |\gamma|$.
  \item \label{item:JKP-3}
    The Koteck\'{y}-Preiss criterion as in \Cref{prop:kp} holds for $g(\cdot)$.
\end{enumerate}
Then there exists an FPTAS for $\Xi(G)$ for every $G \in \mathcal G$.
\end{proposition}

\subsection{Verifying the Koteck\'{y}-Preiss criterion}\label{sec:KP-verify}
The most important condition in \Cref{prop:fptas-JKP20} is the Koteck\'{y}-Preiss criterion as in \Cref{prop:kp}.
Let $g(\gamma)\defeq\rho\abs{\gamma}$ for some sufficiently small constant~$\rho$.
For any $S\subseteq V$, let $N^+(S)$ be the extended neighbourhood of $S$, namely, $N^+(S)\defeq S\cup \partial S$.
Then we have,
\begin{align}
  \sum_{\gamma' \not \sim \gamma}\abs{w(\gamma')}  e^{\abs{\gamma'}+g(\gamma')} & = \sum_{\gamma' \not \sim \gamma}\abs{w(\gamma')}  e^{(1+\rho)\abs{\gamma'}} \notag\\
  & \le \sum_{v\in N^+(\gamma)} \sum_{\gamma'\ni v}\abs{w(\gamma')}  e^{(1+\rho)\abs{\gamma'}}. \label{eqn:KP-weight}
\end{align}
Note that $\abs{N^+(\gamma)}\le (\Delta+1)\abs{\gamma}$.
Moreover, Borgs, Chayes, Kahn, and Lov\'asz showed the following lemma.
\begin{lemma} [\protect{\cite[Lemma 2.1]{BCKL13}}]  \label{lem:BCKL}
  For a graph $G$ of maximum degree $\Delta$,
  the number of connected subgraphs containing $v$ of size~$\ell$ is at most $\frac{1}{\ell}\binom{\ell \Delta}{\ell-1} \le (e\Delta)^{\ell-1}$.
\end{lemma}
Thus we need an upper bound of $w(\gamma)$ for any polymer $\gamma$ of a fixed size $\ell$.

\begin{lemma}  \label{lem:weight-bound}
  For a polymer $\gamma$ of size $\ell$ and $\delta<\frac{1}{\Delta-1}$, 
  \begin{align*}
    \abs{w(\gamma)}\le (e\delta)^{\abs{\gamma}}.
  \end{align*}
\end{lemma}
\begin{proof}
  Note that we can reinterpret \eqref{eqn:polymer-weight} as
  \begin{align*}
    \abs{w(\gamma)} = \Pr\left[~\bigwedge_{v\in \gamma}\overline{C_{v}}~\right],
  \end{align*}
  where the probability is over the product distribution of $\{x_e\}_{e\in E(\gamma)}$ where each $x_e$ is uniform over $\left[ 0,\IntLength \right]$.
  Let $I_\gamma$ be a maximal independent set of $\gamma$.
  Then $\abs{I_{\gamma}}\ge \frac{\abs{\gamma}}{\Delta}$, and
  \begin{align}\label{eqn:prob-indset}
    \Pr\left[~\bigwedge_{v\in \gamma}\overline{C_{v}}~\right] &\le \Pr\left[~\bigwedge_{v\in I_\gamma}\overline{C_{v}}~\right] = \prod_{v\in I_\gamma}\Pr\left[~\overline{C_{v}}~\right]. 
  \end{align}
  Recall that $\overline{C_v}$ is the constraint that $\sum_{e\sim v} x_e>1$.

  If there is some $u\in I_{\gamma}$ such that $\deg_G(u)\le \Delta-1$,
  then $\sum_{e\sim u}x_e\le\frac{\Delta-1}{\Delta}(1+\delta)<1$ as $\delta < \frac{1}{\Delta-1}$ 
  In this case $\Pr\left[~\overline{C_{u}}~\right] = 0$ and the lemma follows.

  Therefore we can assume that $\deg_G(v)=\Delta$ for all $v\in I_{\gamma}$.
  Let $y_e=\IntLength - x_e$.
  Then
  \begin{align*}
    \Pr\left[\sum_{e\sim v} x_e>1\right] & = \Pr\left[ \sum_{e\sim v} y_e<\delta \right],
  \end{align*}
  where each $y_e$ is uniform over $\left[0,\IntLength\right]$.
  Denote by $S_\delta$ the simplex $\left\{\mathbf{y}\mid \sum_{e\sim v}y_e < \delta,\text{ and }\forall e\sim v,~y_e\ge 0\right\}$.
  Then, $S_\delta\subseteq {M_\delta}^{E(v)}$ as $\delta < \IntLength$.
  Thus,
  \begin{align*}
    \Pr\left[ \sum_{e\sim v} y_e<\delta \right] = \frac{\Vol\left( S_\delta \cap {M_\delta}^{E(v)} \right)}{\Vol\left( {M_\delta}^{E(v)} \right)}
    = \frac{\Vol\left( S_\delta \right)}{\Vol\left( {M_\delta}^{E(v)} \right)}
    =\frac{\frac{\delta^{\Delta}}{\Delta!}}{\left( \IntLength \right)^{\Delta}} \le \left( \frac{e\delta}{1+\delta} \right)^{\Delta},
  \end{align*}
  where we used the fact that a standard simplex of dimension $\Delta$ has volume $1/\Delta!$ and $\Delta!\ge \left( \frac{\Delta}{e} \right)^{\Delta}$.
  Plugging it back to \eqref{eqn:prob-indset},
  \begin{align*}
    \abs{w(\gamma)} 
    & \le  \prod_{v\in I_\gamma}\Pr\left[~\overline{C_{v}}~\right] \le \left( \frac{e\delta}{1+\delta} \right)^{\Delta\cdot\frac{\abs{\gamma}}{\Delta}} 
    < \left( e\delta\right)^{\abs{\gamma}}.\qedhere
  \end{align*}
\end{proof}

Now we can verify the Koteck\'{y}-Preiss criterion.

\begin{lemma}  \label{lem:KP-matching}
  For graphs of maximum degree $\Delta\ge 2$,
  $\delta \le \frac{1}{e^{4}\Delta}$, and sufficiently small $\rho$,
  the Koteck\'{y}-Preiss criterion \eqref{eqn:KP} holds.
\end{lemma}
\begin{proof}
  By \eqref{eqn:KP-weight}, \Cref{lem:BCKL}, and \Cref{lem:weight-bound}, for sufficiently small $\rho$,
  \begin{align*}
    \sum_{\gamma \not \sim \gamma'}\abs{w(\gamma')}  e^{\abs{\gamma'}+g(\gamma')} 
    & \le (\Delta+1)\abs{\gamma} \sum_{\ell\ge 1} (e\Delta)^{\ell-1} \left( e\delta \right)^{\ell} e^{(1+\rho)\ell} \\
    & \le \frac{\Delta+1}{e\Delta}\abs{\gamma}\sum_{\ell\ge 1} \left( e^{3+\rho}\delta\Delta \right)^{\ell} 
    \le \frac{1.5}{e}\abs{\gamma}\sum_{\ell\ge 1} \left( e^{\rho-1} \right)^{\ell} \tag{as $\delta \le \frac{1}{e^{4}\Delta}$ and $\Delta\ge 2$}\\
    & \le \left(\frac{1.5}{e}\cdot\frac{e^{\rho-1}}{1-e^{\rho-1}}\right)\abs{\gamma} < \abs{\gamma}.\qedhere
  \end{align*}
\end{proof}

\subsection{Proof of \texorpdfstring{\Cref{thm:FPTAS-matching}}{Theorem 3}}

With the KP criterion verified, we can prove \Cref{thm:FPTAS-matching} now.
\begin{proof}[Proof of \Cref{thm:FPTAS-matching}]
  We apply \Cref{prop:fptas-JKP20}.
  Let $g(\gamma)=\rho\abs{\gamma}$ for a sufficiently small constant $\rho$, which then satisfies Item~\eqref{item:JKP-2}.
  The Item~\eqref{item:JKP-3} follows from \Cref{lem:KP-matching}.
  For Item~\eqref{item:JKP-1},
  determining if $\gamma$ is a polymer is trivial, and computing $g(\gamma)$ is trivial too.

  For computing $w(\gamma)$, recall \eqref{eqn:polymer-weight}.
  It is easy to see that $\left( \IntLength \right)^{\abs{E(\gamma)}} (-1)^{\abs{\gamma}}w(\gamma)$ is in fact the volume of a polytope with variables $\{x_e\}$ for each $e\in E(\gamma)$ and constraints $\overline{C_v}$ for each $v\in \gamma$ and $x_e\ge 0$ for each $e\in E(\gamma)$.
  This polytope has at most $ \binom{\abs{E(\gamma)} + \abs{\gamma}}{\abs{E(\gamma)}} =\binom{\abs{E(\gamma)} + \abs{\gamma}}{\abs{\gamma}}=O((e(\Delta+1))^{\abs{\gamma}})$ vertices,
  as $\abs{E(\gamma)}\le\Delta\abs{\gamma}$.
  Using the polytope volume algorithm by Lawrence \cite{Law91} or by Barvinok \cite{Bar93}, we can thus compute $w(\gamma)$ in $O(e^{C\abs{\gamma}})$ time for some constant $C$ that depends on $\Delta$. 
\end{proof}

\section{Hypergraph matching polytope} \label{sec:hyper-matching}

Now we turn our attention to the hypergraph case and show \Cref{thm:FPTAS-hyper-matching}.
Let $H=(V,E)$ be a hypergraph.
We will assume throughout this section that, for some constants $\Delta,k\ge 2$, vertices of $H$ have maximum degree $\Delta$, and hyperedges of $H$ have maximum size $k$.
Recall that the fractional hypergraph matching polytope is defined in \eqref{eqn:hyper-polytope},
and the truncated version $P_{H,\delta}$ is defined by intersecting $P_H$ with ${M_{\delta}}^E$.

Our $\delta$ will satisfy that $\delta<\frac{1}{\Delta-1}$,
in which case, if the degree of some $v$ is at most $\Delta-1$,
the corresponding constraint is always satisfied.
Thus, we may assume that $\deg_H(v)=\Delta$ for all $v\in V$.

Let $\Flat(H)$ be the flattened graph of $H$, namely its vertex set is still $V$, and for $u,v\in V$, $u$ and $v$ are adjacent if they appear in the same hyperedge $e$ for some $e\in E$.
In other words, we replace each hyperedge $e$ by a clique on the vertex set of $e$.
Similar to \eqref{eqn:swap-integral-prod}, we have
\begin{align}\label{eqn:H-swap-integral-prod}
  \Vol(P_{H,\delta})
  & = \sum_{S\subseteq V}
  \bigintssss_{{M_\delta}^{E\setminus E(S)}}1d\mu
  \left(\prod_{K \in \+K(S)}  (-1)^{\abs{K}} \bigintssss_{{M_\delta}^{E(K)}}\prod _{v \in K}\mathbbm 1_{\overline{C_v}}d\mu\right),
\end{align}
where $\+K(S)$ denote the connected components of the induced subgraph $\Flat(H)[S]$.
Thus, we have a similar polymer model partition function
\begin{align}  \label{eqn:H-pf-polymer}
  \Xi(H) = \frac{\Vol(P_{H,\delta})}{\left( \IntLength \right)^{\abs{E}}}= \sum_{S\subseteq V}\prod_{K \in \+K(S)} w_H(K),
\end{align}
where 
\begin{align} \label{eqn:H-polymer-weight}
  w_H(K) = \left( \IntLength \right)^{-\abs{E(K)}} (-1)^{\abs{K}} \bigintssss_{{M_\delta}^{E(K)}} \prod _{v \in K}\mathbbm 1_{\overline{C_v}}d\mu.
\end{align}
While \eqref{eqn:H-pf-polymer} and \eqref{eqn:H-polymer-weight} look identical to \eqref{eqn:pf-polymer} and \eqref{eqn:polymer-weight},
the underlying graph is different and consequently \Cref{lem:weight-bound} no longer applies.
We may view this as a polymer model over $\Flat(H)$,
by treating connected induced subgraphs of $\Flat(H)$ as polymers.

Directly applying the method of \Cref{lem:weight-bound} would yield
\begin{align*}
  \abs{w_H(\gamma)} 
  & < \left( e\delta\right)^{\frac{\abs{\gamma}}{k-1}}.
\end{align*}
This bound can only validates the Koteck\'{y}-Preiss criterion up to $\delta = \Theta\left(\frac{1}{\Delta k}\right)^{k-1}$. 
To get a better bound, we need a different polymer model.

\subsection{A different polymer model}

Let $\Inc(H)$ be the incidence graph of $H$, where the vertex set is $V\cup E$, and for $v\in V$ and $e\in E$, $v\sim e$ in $\Inc(H)$ if $v\in e$ in $H$.
Then, $\Inc(H)$ is a bipartite graph, and $\Flat(H)$ is exactly $\Inc(H)$ projected on the $V$ side.
For each subset $S\subset V$ of vertices, let $\partial_{\Inc(H)} S$ be the neighbours of $S$ in $\Inc(H)$.
Note that $\partial_{\Inc(H)}S=E(S)$, namely the set of hyperedges that contain at least one vertex in $S$.

We rewrite \eqref{eqn:H-pf-polymer} into a different polymer model.
For a connected induced subgraph $K$ of $\Flat(H)$,
we map it to a \emph{minimal connected subgraph} (MCS) $T$ as follows.
Fix an arbitrary ordering of $V$.
We consider vertices of $K$ in order one at a time, to potentially add them to $T$.
Initialise $T$ to be the first vertex of $K$ and remove it from $K$.
Given the current $T$, we find the first vertex $u$ in $K$ that is adjacent to the current $T$ in $\Flat(H)$.
Add $u$ to $T$ if $E(v_i)\not\subseteq E(T)$.
In other words, we only add $v_i$ if it introduces a new hyperedge to $E(T)$.
Whether $u$ is added to $T$ or not, we remove $u$ from $K$. 
We keep doing this until $K$ is empty.
Denote this mapping by $\varphi$.
We also extend the mapping to any (not necessarily connected) subset $S\subseteq V$ by defining $\varphi(S)=\emptyset$ if $\Flat(H)[S]$ is not connected.
We call a non-empty image $T$ of $\varphi$ an MCS,
namely, $T$ is an MCS if there is a connected $K\subseteq V$ such that $\varphi(K)=T$.
Equivalently, a connected $T$ is an MCS if and only if $\varphi(T)=T$.

\begin{lemma}  \label{lem:MCS}
  Let $T$ be an MCS.
  Then,
  \begin{enumerate}
    \item $\abs{E(T)}\ge\abs{T}$;
    \item $T$ is connected in $\Flat(H)$.
  \end{enumerate}
\end{lemma}
\begin{proof}
  Both claims are straightforward by an induction on the number of vertices of $T$.
\end{proof}

While an MCS is connected in $\Flat(H)$, it is not necessarily a tree.
Indeed, an MCS is determined by $\Inc(H)$ and not by $\Flat(H)$ alone.

Our new model have MCSes as polymers.
As before two MCSes are incompatible if their distance is at most $1$.
For this model, define a new weight function
\begin{align*}
  u(T)\defeq \sum_{K:\varphi(K)=T} w_H(K).
\end{align*}
We rewrite \eqref{eqn:H-pf-polymer} as 
\begin{align}  \label{eqn:H-pf-polymer-2}
  \Xi(H) = \sum_{F}\prod_{T \in \+K(F)} u(F),
\end{align}
where the sum is over $F\subseteq V$ whose every component is an MCS.

For an MCS $T$, a vertex $v\in V\setminus T$ is said to be \emph{broken} if $\varphi(T\cup\{v\})=T$.
Let $B(T)$ be the set of broken vertices with respect to~$T$.
This notion is a hypergraph generalisation of the well-known broken circuit theory for graphs \cite{Whi32,Tut54}.

\begin{lemma}  \label{lem:MCS-broken}
  For $K\subseteq V$, $\varphi(K)=T$ if and only if $T\subseteq K\subseteq (T\cup B(T))$.
\end{lemma}
\begin{proof}
  If $\varphi(K)=T$,
  the inclusion $T\subseteq K$ follows directly from the definition of $\varphi$.
  For the other inclusion, we induce on the size of $K$.
  If $K=T$, the claim holds trivially.
  For the inductive step, let $v\in K\setminus T$ be the first vertex that is not added to $T$ during $\varphi$.
  Then, clearly $\varphi(K\setminus\{v\})=T$, and thus by the induction hypothesis $K\setminus\{v\}\subseteq(T\cup B(T))$.
  We claim that $v\in B(T)$ as well.
  This is because the mapping $\varphi(T\cup\{v\})$ would behave the same as $\varphi(K)$ until $v$ is considered (and then dropped),
  as before $v$ all vertices considered in $\varphi(K)$ are added to $T$.
  As $v$ is dropped next, the rest of the process $\varphi(T\cup\{v\})$ just processes vertices of $T$ and add them.
  It implies that $\varphi(T\cup\{v\})=T$, and thus $v\in B(T)$.
  This finishes the inductive step and shows that $K\subseteq (T\cup B(T))$.
  
  For the other direction, suppose $T\subseteq K\subseteq (T\cup B(T))$.
  We also do an induction on $\abs{K}$.
  The base case of $K=T$ is trivial.
  For the inductive step, consider the first $v\in K\setminus T$ that is processed during $\varphi(K)$.
  Note that up to $v$ all processed vertices are in $T$, and the process is identical to the process of $\varphi(T\cup\{v\})$ up to this point.
  Thus, $v$ would be discarded then.
  The rest of $\varphi(K)$ is identical to $\varphi(K\setminus\{v\})$,
  and by the induction hypothesis, $\varphi(K\setminus\{v\})=T$.
  Thus, $\varphi(K)=T$.
\end{proof}

Notice that if $\varphi(K)=T$, $E(K)=E(T)$.
We then have a better expression for $u(T)$ as follows,
\begin{align}
  u(T) & = \sum_{K:\varphi(K)=T} w_H(K) = \sum_{K:\varphi(K)=T} \left( \IntLength \right)^{-\abs{E(T)}} (-1)^{\abs{K}} \bigintssss_{{M_\delta}^{E(K)}} \prod _{v \in K}\mathbbm 1_{\overline{C_v}}d\mu \tag{by \eqref{eqn:H-polymer-weight}}\\
  & = \left( \IntLength \right)^{-\abs{E(T)}} \sum_{K:T\subseteq K\subseteq T\cup B(T)}  (-1)^{\abs{K}} \bigintssss_{{M_\delta}^{E(T)}} \prod_{v \in K}\mathbbm 1_{\overline{C_v}}d\mu \tag{by \Cref{lem:MCS-broken}}\\
  & = \left( \IntLength \right)^{-\abs{E(T)}} (-1)^{\abs{T}}\sum_{A\subseteq B(T)}   \bigintssss_{{M_\delta}^{E(T)}} \prod_{v \in T}\mathbbm 1_{\overline{C_v}}
    \prod_{v \in A}-\mathbbm 1_{\overline{C_v}}d\mu \notag\\ 
  & = \left( \IntLength \right)^{-\abs{E(T)}} (-1)^{\abs{T}} \bigintssss_{{M_\delta}^{E(T)}} \prod_{v \in T}\mathbbm 1_{\overline{C_v}}
    \prod_{v \in B(T)}\left( 1-\mathbbm 1_{\overline{C_v}} \right)d\mu \notag\\
  & = \left( \IntLength \right)^{-\abs{E(T)}} (-1)^{\abs{T}} \bigintssss_{{M_\delta}^{E(T)}} \prod_{v \in T}\mathbbm 1_{\overline{C_v}}
  \prod_{v \in B(T)}\mathbbm 1_{C_v}d\mu. \label{eqn:uT-sum}
\end{align}

Next we give an upper bound for $\abs{u(T)}$.

\begin{lemma}  \label{lem:MCS-weight}
  Let $T$ be a MCS of size $\ell$. 
  For $\delta<\frac{1}{\Delta}$,
  \begin{align*}
    \abs{u(T)}\le \left( \Delta \delta \right)^{\ell} \left( \frac{e}{\Delta} \right)^{\frac{\ell}{k-1}}.
  \end{align*}
\end{lemma}
\begin{proof}
  By \eqref{eqn:uT-sum},
  \begin{align*}
    \abs{u(T)} & \le \left( \IntLength \right)^{-\abs{E(T)}} \bigintssss_{{M_\delta}^{E(T)}} \prod_{v \in T}\mathbbm 1_{\overline{C_v}}d\mu.
  \end{align*}
  Thus,
  \begin{align*}
    \abs{u(T)}\le \Pr\left[ \bigwedge_{v \in T}\overline{C_v} \right],
  \end{align*}
  where the probability is over the product distribution of $x_e$'s for $e\in E(T)$, and each $x_e$ is uniform over $\left[ 0,\IntLength \right]$.
  Notice that $\overline{C_v}$ means $\sum_{e\sim v} x_e > 1$,
  which is equivalent to $\sum_{e\sim v} y_e < \delta$ where $y_e=\IntLength-x_e$. 
  As $y_e\ge 0$, this implies that $y_e<\delta$ for all $e\sim v$.
  Since $T\cup E(T)$ is connected in $\Inc(H)$, 
  if $\prod_{v \in T}\mathbbm 1_{\overline{C_v}}=1$,
  then for all $e\in E(T)$, $y_e<\delta$.

  Moreover, for a maximal independent set $I_T$ of $T$ (in $\Flat(H)$), we have that $\abs{I_T}\ge \frac{\abs{T}}{\Delta(k-1)}$,
  as the maximum degree of $\Flat(H)$ is $\Delta(k-1)$.
  Note that for $v\in I_T$, the constraints $C_v$ are independent from each other.
  Then we have 
  \begin{align*}
    \Pr\left[ \bigwedge_{v \in T}\overline{C_v} \right] & \le \Pr\left[ \bigwedge_{v \in I_T}\overline{C_v}\wedge\bigwedge_{e\in E(T)\setminus E(I_T)} y_e<\delta \right] \\
    & = \prod_{v \in I_T}\Pr\left[\;\overline{C_v}\; \right]\prod_{e\in E(T)\setminus E(I_T)}\Pr\left[  y_e<\delta \right]\\
    & = \left( \IntLength \right)^{-\abs{E(T)}} \delta^{\abs{E(T)}-\abs{E(I_T)}} \left( \frac{\delta^{\Delta}}{\Delta!} \right)^{\abs{I_T}},
  \end{align*}
  where we use the fact that the volume of a standard simplex of dimension $\Delta$ is $\frac{1}{\Delta!}$.
  (Recall that we can assume that all vertices have degree $\Delta$, as vertices of smaller degrees correspond to a constraint that trivially holds.)
  As $I_T$ is an independent set, $\abs{E(I_T)}=\Delta\abs{I_T}$.
  Together with $\Delta!\ge \left( \frac{\Delta}{e} \right)^{\Delta}$ and $\abs{I_T}\ge \frac{\abs{T}}{\Delta(k-1)}$, we have
  \begin{align*}
    \Pr\left[ \bigwedge_{v \in T}\overline{C_v} \right] \le \left( \Delta \delta \right)^{\abs{E(T)}} \left( \left( \frac{e}{\Delta} \right)^{\Delta} \right)^{\frac{\abs{T}}{\Delta(k-1)}}.
  \end{align*}
  As $T$ is a MCS, by \Cref{lem:MCS},
  $\abs{E(T)}=\abs{\partial_{\Inc(H)}T}\ge \abs{T}= \ell$.
  Then, as $\delta\Delta<1$,
  \begin{align*}
    \abs{u(T)} & \le \left( \Delta \delta \right)^{\ell} \left( \frac{e}{\Delta} \right)^{\frac{\ell}{k-1}}.\qedhere
  \end{align*}
\end{proof}

Given \Cref{lem:MCS-weight}, the verification of the Koteck\'{y}-Preiss criterion \eqref{eqn:KP} is very similar to previous arguments.

\begin{lemma}  \label{lem:KP-hyper-matching}
  For hypergraphs of maximum degree $\Delta\ge 2$ and maximum hyperedge size $k\ge 2$,
  $\delta \le \left( e^{4}\Delta^{\frac{2k-3}{k-1}}(k-1) \right)^{-1}$, and sufficiently small $\rho$,
  the Koteck\'{y}-Preiss criterion \eqref{eqn:KP} holds for MCSes and $u(\cdot)$.
\end{lemma}
\begin{proof}
  By \Cref{lem:MCS}, every MCS is a connected subgraph in $\Flat(H)$.
  Thus, by \Cref{lem:BCKL}, the number of MCSes of size $\ell$ is at most $(e\Delta(k-1))^{\ell}$.
  For $\Delta\ge 2$, $k\ge 2$, and sufficiently small constant $\rho$,
  \begin{align*}
    \sum_{T' \not \sim T}\abs{w(T')}  e^{\abs{T'}+g(T')} 
    & \le (\Delta(k-1)+1)\abs{T} \sum_{\ell\ge 1} (e\Delta(k-1))^{\ell-1} \left( \Delta \delta \right)^{\ell} \left( \frac{e}{\Delta} \right)^{\frac{\ell}{k-1}} e^{(1+\rho)\ell}\\
    & = \frac{\Delta(k-1)+1}{e\Delta(k-1)}\abs{T}\sum_{\ell\ge 1} \left( e^{2+\frac{1}{k-1}+\rho}\delta\Delta^{\frac{2k-3}{k-1}} (k-1) \right)^{\ell} 
    \le \frac{1.5}{e}\abs{T}\sum_{\ell\ge 1} \left( e^{\rho-1} \right)^{\ell} \\
    & \le \left(\frac{1.5}{e}\cdot\frac{e^{\rho-1}}{1-e^{\rho-1}}\right)\abs{T} < \abs{T}.\qedhere
  \end{align*}
\end{proof}

\subsection{Proof of \texorpdfstring{\Cref{thm:FPTAS-hyper-matching}}{Theorem 4}}

With the KP criterion verified, we can prove \Cref{thm:FPTAS-hyper-matching} now.
\begin{proof}[Proof of \Cref{thm:FPTAS-hyper-matching}]
  We apply \Cref{prop:fptas-JKP20} on $\Flat(H)$.
  For a polymer (MCS) $T$,
  let $g(T)=\rho \abs{T}$ for a sufficiently small constant $T$, which then satisfies Item~\eqref{item:JKP-2}.
  Item~\eqref{item:JKP-3} follows from \Cref{lem:KP-hyper-matching}.
  For Item~\eqref{item:JKP-1}, computing $g(T)$ is trivial.
  To determine if $T$ is a polymer is trivial, we just check if $\varphi(T)=T$.
  This takes time linear in $\abs{T}$.
  
  For computing $u(T)$,
  we use the expression \eqref{eqn:uT-sum}.
  The integral in \eqref{eqn:uT-sum} is the volume of a polytope defined by the constraints $\overline{C_v}$ for $v\in T$, $C_v$ for $v\in B(T)$, and $x_e\ge 0$ for $e\in E(T)$.
  To find the defining constraints of this polytope, we need to compute $B(T)$.
  Note that by definition, $v\in B(T)$ only if $T\cup\{v\}$ is connected,
  which implies that $B(T)\subseteq \partial_{\Flat(H)}T$.
  Then we just need to go through every $v\in \partial_{\Flat(H)}T$ and check if $\varphi(T\cup \{v\})=T$.
  As $\abs{\partial_{\Flat(H)}T} \le \Delta (k-1) \abs{T}$,
  this takes at most $O(\abs{T}^2)$ time.

  There are $\abs{T}+\abs{B(T)}+\abs{E(T)}$ constraints and $\abs{E(T)}$ variables.
  Note that $\abs{B(T)}\le\abs{\partial_{\Flat(H)}T}\le \Delta (k-1) \abs{T}$.
  Recall that $\abs{E(T)}\le \Delta \abs{T}$ and by \Cref{lem:MCS}, $\abs{E(T)}\ge\abs{T}$.
  Thus, the number of vertices of this polytope is at most
  \begin{align*}
    \binom{\abs{T}+\abs{B(T)}+\abs{E(T)}}{\abs{E(T)}} & \le \left(\frac{e(\abs{T}+\abs{B(T)}+\abs{E(T)})}{\abs{E(T)}}\right)^{\abs{E(T)}}\\
    & \le (e\Delta k +e)^{\Delta\abs{T}}.
  \end{align*}
  Using the polytope volume algorithm by Lawrence \cite{Law91} or by Barvinok \cite{Bar93}, we can compute $u(T)$ in $O(e^{C\abs{\gamma}})$ time for some constant $C$ that depends on $\Delta$ and $k$. 
  This verifies Item~\eqref{item:JKP-1} of \Cref{prop:fptas-JKP20} and finishes the proof.
\end{proof}

\section{Concluding remarks} \label{sec:discuss}

Our work leaves open the question of whether, under the setting of \Cref{thm:FPTAS-hyper-matching}, an FPTAS exists for $\delta=\frac{C}{\Delta k}$ for some constant $C>0$.
The most straightforward potential approach to this bound is to strengthen the bound in \Cref{lem:MCS-weight} to $\left( C\delta \right)^{\abs{T}}$.
However, this strengthening does not appear to hold, at least for some parameters $k$ and $\Delta$, when $\delta=\frac{C}{\Delta k}$.
To see this, consider $T=\{v_1,\dots,v_k\}$ of size $k$.
Let the hyperedges $e_1=\dots=e_{\Delta-1}=T$.
Moreover, for each $i\in[k]$, introduce $f_i$ such that $v_i\in f_i$ but $v_j\not\in f_i$ for any $j\neq i$. 
Namely, each $f_i$ contains only $v_i$ from $T$ (and possibly other vertices from outside of $T$).
Note that $T$ is an MCS and let $B(T)=\emptyset$.
Then, $\abs{T}=k$ and $\abs{E(T)}=\Delta-1+k$.
We have
\begin{align*}
  \bigintssss_{{M_\delta}^{E(T)}} \prod_{v \in T}\mathbbm 1_{\overline{C_v}} d\mu &= \sum_{i=1}^k\bigintssss_{{M_\delta}^{E(T)}} \mathbbm 1_{\overline{C_{v_i}}}\prod_{j=1}^{k}  \mathbbm 1_{x_{f_j}\ge x_{f_i}} d\mu\\
  & = \sum_{i=1}^k\bigintssss_{{M_\delta}^{E(T)}} \mathbbm 1_{\sum_{t=1}^{\Delta-1}y_{e_t}\le\delta-y_{f_i}}\prod_{j=1}^{k}  \mathbbm 1_{y_{f_j}\le y_{f_i}} d\mu\\
  & = \sum_{i=1}^k\bigintssss_{0}^{\delta} \frac{(\delta-y_{f_i})^{\Delta-1}y_{f_i}^{k-1} }{(\Delta-1)!}d y_{f_i}\\
  & = \frac{k!\delta^{\Delta-1+k}}{(\Delta-1+k)!},
\end{align*}
where we substitute $y_{e}=\IntLength-x_e$ for all $e\in\{e_1,\dots,e_{\Delta-1},f_1,\dots,f_{k}\}$.
Then, 
\begin{align*}
  \abs{u(T)}& = \left( \IntLength \right)^{-(\Delta-1+k)}\frac{k!\delta^{\Delta-1+k}}{(\Delta-1+k)!} \\
  & \ge \left( \IntLength \right)^{-(\Delta-1+k)}\frac{\delta^{\Delta-1+k}}{(\Delta-1+k)^{\Delta-1}}\\
  & = (1+\delta)^{-(\Delta-1+k)} \left( \frac{\Delta \delta}{\Delta-1+k} \right)^{\Delta-1}(\Delta\delta)^{k}.
\end{align*}
Plug in $\delta=\frac{C}{\Delta k}$ and let $k\ge \Delta$ such that $\frac{k}{\log k}\gg\frac{\Delta}{\log \Delta}$ (for example $k=\Delta^2$).
Then,
\begin{align*}
  \frac{\abs{u(T)}}{\delta^k} 
  & \ge C_0 \left( \frac{C}{2k^2} \right)^{\Delta-1}\Delta^{k}\\
  & \ge C_1^k = C_1^{\abs{T}},
\end{align*}
for some constant $C_0>0$ and any constant $C_1>0$ when $\Delta\rightarrow\infty$.
Thus, to achieve an FPTAS for $\delta=\frac{C}{\Delta k}$, some new ideas are required.

\bibliographystyle{alpha}
\bibliography{volume-bibliography.bib}

\end{document}